\documentclass{ifacconf}
\usepackage{graphicx}      
\usepackage{natbib}  
\usepackage{amsmath,mathtools,amssymb}

\newenvironment{proof}[1][Proof]{\textbf{#1.} }{\ \rule{0.5em}{0.5em}}

\usepackage{url}

\usepackage{amsfonts}
\usepackage{graphicx}
\usepackage{color}
\usepackage{float}
\usepackage{tabularx}

\usepackage{gensymb}
\usepackage{amssymb}
\usepackage{array}
\usepackage{amsmath,mathtools}
\usepackage{bm}
\usepackage{mathrsfs}

\usepackage{graphicx}
\usepackage{epsfig}
\usepackage{subfig}
\usepackage{float}

\usepackage{multirow}
\usepackage{longtable}
\usepackage{booktabs}
\usepackage{tikz,pgfplots,pgfmath} 
\usetikzlibrary{math,positioning,intersections,shapes.misc,shapes.geometric,fit,arrows,arrows.meta,decorations.pathmorphing,decorations.pathreplacing,decorations.shapes,decorations.markings,patterns,calc,shadings}

\newtheorem{theorem}{Theorem}

\newtheorem{remark}{Remark}[section]

\newtheorem{assumption}{Assumption}

\begin{document}
\begin{frontmatter}
\title{ Finite-Time Control Based on Differential Flatness for Wheeled Mobile Robots with Experimental Validation}
\author[ImViA,COMSATS]{Imtiaz Ur Rehman}
\author[LIS]{Moussa Labbadi}
\author[ImViA]{Amine Abadi}
\author[ImViA]{Lew Lew Yan Voon}

\thanks[footnoteinfo]{We thank the French government for the Plan France Relance initiative which provided fundings via the European Union under contract ANR-21-PRRD-0047-01. We are also grateful to the Doctoral School and the French Ministry of Research for the PhD MENRT scholarship.}

\address[ImViA]{ImViA Laboratory EA 7535, Université Bourgogne Europe,  
720 avenue de l'Europe, 71200 Le Creusot, France  
(e-mail: Imtiaz-Ur.Rehman@ube.fr; Amine.Abadi@ube.fr; lew.lew-yan-voon@ube.fr)}

\address[COMSATS]{COMSATS University Islamabad, Islamabad, Pakistan  
(e-mail: imtiaz.rehman@comsats.edu.pk)}

\address[LIS]{Aix-Marseille University, LIS UMR CNRS 7020, Marseille, France  
(e-mail: moussa.labbadi@lis-lab.fr)}

\begin{abstract}   
A robust tracking control strategy is designed to empower wheeled mobile robots (WMRs) to track predetermined routes while operating in diverse fields and encountering disturbances like strong winds or uneven path conditions, which affect tracking performance. Ensuring the applicability of this tracking method in real-world scenarios is essential. To accomplish this, the WMR model is initially transformed into a linear canonical form by leveraging the differential flatness of its kinematic model, facilitating controller design. Subsequently, a novel integral nonlinear hyperplane-based sliding mode control (INH-SMC) technique is proposed for WMR under disturbances. The stability of the technique is analyzed and verified. Finally, its practical viability is demonstrated through a comparative real-world indoor experiment on a TurtleBot3 WMR subjected to disturbances, confirming the feasibility and efficacy of the proposed approach.
\end{abstract}


\end{frontmatter}

\section{INTRODUCTION}

Wheeled mobile robots (WMRs) are rapidly expanding across various domains. Technological advancements have enabled these robots to move beyond laboratories into numerous industries, including military, mining, transportation, manufacturing, exploration, agriculture, and healthcare~\citep{ebel2024cooperative,yepez2023mobile,garaffa2021reinforcement}. WMRs perform diverse tasks such as autonomous navigation, surveillance, material handling, and inspection, even in domestic settings like automated cleaning, door control, and food serving~\citep{zhao2021wheeled,al2021embedded,hassan2024path}. Since WMRs are used in various fields, the primary objective of most robot applications is trajectory tracking (TT), which requires robots to follow predetermined paths or trajectories. However, when WMRs operate outdoors, they confront various disturbances that affect accurate TT, including model uncertainties, uneven terrain, and external disturbances such as wind.

As a result, TT has become an important area of study in robotics. Numerous studies have explored this topic using diverse approaches~\citep{zangina2020non,yousuf2021dynamic,xu2020enhanced,tang2010differential}. Nevertheless, these solutions fail to account for disturbances. Thus, to ensure robustness in TT, scholars have put forth several techniques. Neural network (NN)-based methods have tackled the dual objectives of TT and robustness in~\citep{korayem2024adaptive,chen2020adaptive,li2018adaptive}. Furthermore, fuzzy algorithms are illustrated in~\citep{tolossa2024trajectory, moudoud2022fuzzy, chen2023nonlinear, chwa2011fuzzy, singhal2022robust}. In addition, because differential flatness (DF) excels at simplifying the control design process, it has become more prevalent among diverse control strategies. The flatness condition enables the system's inputs and states to be expressed in terms of independent outputs, known as flat outputs, and their derivatives~\citep{fliess1995flatness}. In the works of Yuan et al.~\citep{yuan2024differential} and Liu et al.~\citep{liu2023autonomous}, robust tracking control (RTC) algorithms were developed for WMRs employing the DF framework to meet the dual goals. 

 One of the several control algorithms used for the TT of a WMR under disturbances is sliding mode control (SMC). SMC-based TT control algorithms are proposed to fix the robustness problem and are detailed in~\citep{mera2020sliding, yang2018sliding, nath2021event, mu2017nonlinear}. likewise, as described in~\citep{dang2023adaptive,wu2019backstepping}, SMC techniques based on NN and fuzzy logic are also proposed to handle the TT problem together with robustness. Fuzzy and NN-based control algorithms have become more common in the reviewed study due to their ability to achieve TT and their exceptional versatility in resolving disturbance issues. However, they do have certain disadvantages, including demanding processing requirements and complicated structural designs. Moreover, the dual objectives are likewise addressed by SMC algorithms, but they do not exhibit finite-time convergence (FTC). Therefore, terminal SMC (TSMC) has been studied in~\citep{benaziza2017mobile}. Despite achieving FTC, terminal SMC suffers from slow convergence and a singularity issue. To resolve this, integral TSMC (ITSMC) is presented in~\citep{sun2021trajectory}. Additionally, fast nonsingular TSMC (FNTSMC), an enhanced variant of TSMC, has been introduced in~\citep{xie2021robust}. Thus, this work proposes a novel hyperplane SMC technique using DF for the WMR exposed to external disturbances by combining the notable benefits of integral TSMC and FNTSMC. The following highlights our core contributions in this paper: i) By leveraging the DF property, the nonlinear kinematic model of the WMR is transformed into a linear structure, effectively addressing underactuation and converting nonmatching disturbances into matching ones; ii) A novel RTC law is proposed for the TT problem of a WMR that ensures precise tracking even when exposed to disturbances. In this approach, an integral nonlinear hyperplane-based sliding manifold is introduced, designed with an NTSM surface to provide robustness, precise tracking, and fast convergence; iii) A switching control law, synthesized based on the proposed sliding manifold, is employed to counteract the upper bound of strong disturbances; iv) The practicality and effectiveness of the proposed controller are demonstrated through experimental validation on a TurtleBot3 WMR under disturbances, with a comparison to existing methods in terms of robustness, TT performance and control effort. A video of the experiment is provided via the link in the Results section.

 The following is an outline of the article:
Section II sets forth the mathematical model for the WMR. Section III illustrates the proposed RTC algorithm. The related experimental results and conclusions are summarized in Sections IV and V, respectively.

\section{Mathematical Model}
A two-wheeled differential-drive mobile robot is the subject of this study. The generalized WMR setup is given by $\Omega = [x\; y\; \theta]^\top$, where $(x,y)$ represents the Cartesian coordinates of the center, and $\theta$ indicates the orientation angle. The WMR has two wheels, each with a radius $r$, separated by a distance $D$. The angular velocities of the left and right wheels are represented by $u_l$ and $u_r$, respectively. The following expressions define the rotational velocity $w$ and translational velocity $v$ of the WMR: 
\begin{align}\label{eq3}
v&=\frac{r}{2}(u_l+u_r), \nonumber \\
w&=\frac{r}{D}(u_r-u_l).
\end{align}
The kinematic model (KM) of a WMR is formulated as follows:
\begin{align} \label{eqrobot}
\dot \Omega&=[v\cos\theta\; v\sin\theta\; w]^\top.
\end{align}
The WMR presents a control challenge since it has three states and two inputs. Reducing the system’s dimensionality is necessary to facilitate control design for a system with DF. To show the DF of the WMR, the flat outputs (FOs) are defined as follows:
\begin{align}\label{eq:FO1}
   \Gamma&=[\Gamma_{11}\;  \Gamma_{21}]^\top = [x\;  y]^\top. 
\end{align}
Given (\ref{eq:FO1}), its time derivative yields:
\begin{align}\label{eq:FO2}
   \dot\Gamma&=\begin{bmatrix}
\dot{\Gamma}_{11} \\ \dot{\Gamma}_{21} \\
\end{bmatrix}
=
\begin{bmatrix} 
	\dot x \\ \dot y
\end{bmatrix} =
\begin{bmatrix} 
	\cos \theta \\ \sin \theta
\end{bmatrix}v.
\end{align}
Computing the time derivative of (\ref{eq:FO2}) results in:

\begin{align}\label{eq:FO3}
   \ddot\Gamma&=\begin{bmatrix}
\ddot{\Gamma}_{11} \\ \ddot{\Gamma}_{21} \\
\end{bmatrix}
=
\begin{bmatrix} 
	\ddot x \\ \ddot y
\end{bmatrix} =\underbrace{
\begin{bmatrix}
       \cos\theta & -v \sin\theta  \\ 
       \sin\theta & v \cos\theta
\end{bmatrix}}_{N}\begin{bmatrix}u_{n1} \\ u_{n2}\end{bmatrix},
\end{align}
where $u_{n1}$ and $u_{n2}$ depict the new inputs, with $u_{n1} = \dot{v}$ and $u_{n2} = \dot{\theta}$. The following outlines how the system's states and inputs can be characterized in terms of the FOs and their derivatives:
\begin{align}\label{eq:FO4}
   x&=\Gamma_{11}, \quad y=\Gamma_{21}, \quad \theta=\tan^{-1} \left(\frac{\dot \Gamma_{21}}{\dot \Gamma_{11}}\right).
\end{align}
\begin{align}\label{eq:FO5}
   v&=\sqrt{\dot \Gamma_{11}^{2}+\dot \Gamma_{21}^{2}}, \quad w=\frac{\dot \Gamma_{11} \ddot \Gamma_{21}-\ddot \Gamma_{11} \dot \Gamma_{21}}{\dot \Gamma_{11}^2+\dot \Gamma_{21}^2}.
\end{align}
Using the given FOs, the new inputs $u_{n1}$ and $u_{n2}$ can be outlined as below:
\begin{align}\label{eq:FO6}
   u_{n1}&=\dot v=\frac{\dot \Gamma_{11} \ddot \Gamma_{11}+\ddot \Gamma_{21} \dot \Gamma_{21}}{\sqrt{\dot \Gamma_{11}^{2}+\dot \Gamma_{21}^{2}}}, \nonumber\\
   u_{n2}&=w=\frac{\dot \Gamma_{11} \ddot \Gamma_{21}-\dot \Gamma_{21} \ddot \Gamma_{11} }{\dot \Gamma_{11}^2+\dot \Gamma_{21}^2}.
\end{align}
Contrary, the system's states and inputs can potentially fully represent the FOs:
\begin{align}\label{eq:FO7}
   \Gamma_{11}&=x, \quad \Gamma_{21}=y, \nonumber\\
   \dot{\Gamma}_{11}&= v \cos\theta, \quad \dot{\Gamma}_{21}= v \sin\theta, \nonumber\\
   \ddot{\Gamma}_{11}&= u_{n1}\cos\theta -u_{n2}v \sin\theta, \nonumber\\ 
    \ddot{\Gamma}_{21}&= u_{n1}\sin\theta +u_{n2}v \cos\theta.
\end{align}
To summarize, the WMR system is DF. Full state controllability within the FO space is achievable when the number of FOs ($\Gamma$) equals the number of inputs~\citep{fliess1995flatness}.

The matrix $N$ in (\ref{eq:FO3}) is non-singular if $v \neq 0$. In this context, control can be defined as follows:
\begin{align}\label{eq:FO8}
 \begin{bmatrix}
u_{n1} \\ u_{n2} 
\end{bmatrix}
&=
N^{-1}\begin{bmatrix}
\ddot{\Gamma}_{11} \\ \ddot{\Gamma}_{21} \\
\end{bmatrix}.   
\end{align}
Substituting (\ref{eq:FO8}) into (\ref{eq:FO3}) results in a linearized system in Brunovsky Form (BF), with the following representation:
\begin{gather*}
\dot \Gamma_{11} =\Gamma_{12}, \quad 
\dot \Gamma_{12} =v_x, \quad  
\dot  \Gamma_{21}= \Gamma_{22}, \quad 
\dot  \Gamma_{22}=v_y.
\end{gather*}

Subject to ideal conditions, the suitable virtual feedback control laws $v_x$ and $v_y$ can adequately follow the desired trajectories $\Gamma_{xd}$ and $\Gamma_{yd}$ for the FOs $\Gamma_{11}$ and $\Gamma_{21}$, respectively. However, external perturbations, like uneven terrain or wind, are not accounted for in the KM of the WMR. Moreover, nominal control laws alone are insufficient to yield good results in practical situations due to the presence of disturbances. To address these effects, Section \ref{sec:cd} proposes a RTC law for the WMR.
\subsection{WMR Model With Disturbances}\label{sec: mm_2}
The WMR model, which is susceptible to external disturbances, is formulated as follows:
\begin{align}\label{UKM}
   \dot{\Omega}&=
\begin{bmatrix}
       \cos\theta & 0  \\ 
       \sin\theta & 0  \\ 
       0 & 1 
\end{bmatrix}\begin{bmatrix}v \\ w\end{bmatrix} + \bar{d},
\end{align}
where the external perturbations impacting the WMR are illustrated by $\bar{d}=[d_x\; d_y\; d_\theta]^\top$. A differentially flat depiction that accounts for these disturbances can be derived using the model in (\ref{UKM}) as outlined:
\begin{small}
\begin{align}\label{UKM_1} 
\begin{bmatrix}
     \ddot{\Gamma}_{11} \\ \ddot{\Gamma}_{21}
\end{bmatrix}
&=\begin{bmatrix}
       \cos\theta & -v \sin\theta  \\ 
       \sin\theta & v \cos\theta
\end{bmatrix}
\begin{bmatrix} 
	 u_{n1} \\ u_{n2} 
\end{bmatrix} + \underbrace{\begin{bmatrix}
\dot d_x -v d_\theta \sin\theta\\ 
      \dot d_y +v d_\theta \cos\theta   
\end{bmatrix}}_{\bar{D}}.
\end{align}
\end{small}
Inserting (\ref{eq:FO8}) into (\ref{UKM_1}) gives the following result:
\begin{equation}\label{eqlambda}
\ddot \Gamma=v+\varpi,
\end{equation}
whereas $\ddot \Gamma$=$[\ddot \Gamma_{11},\ddot \Gamma_{21}]^\top$, $v$=$[v_x\; v_y]^\top$ and $\varpi$=$[\varpi_x\; \varpi_y]^\top$=$\bar{D}$.
Note that disturbances are denoted by $\varpi$. Reformulating (\ref{eqlambda}) in the context of two-linear integrator systems exposed to disturbances leads to the following formulations:
\begin{gather*}
\dot \Gamma_{11} = \Gamma_{12},\; 
\dot \Gamma_{12} = v_x + \varpi_x,\; 
\dot \Gamma_{21} = \Gamma_{22},\; 
\dot \Gamma_{22} = v_y + \varpi_y.
\end{gather*}
\begin{assumption}\label{ass:1}
The disturbances applied to the WMR are unknown but bounded. 
\end{assumption}
\begin{remark}
Assumption \ref{ass:1} holds in practice; extrinsic perturbations, such as rough or uneven surfaces and wind, can have a detrimental impact on the TT performance of WMRs. Nevertheless, these factors have a limited effect on position variation.
\end{remark}

\section{Robust Control Design}\label{sec:cd}
To fulfill the previously stated control objectives for the WMR in the presence of disturbances, a robust tracking control (RTC) law is developed. The approach introduces sliding mode surfaces and constructs novel hyperplane-based sliding manifolds to enhance robustness and convergence. Subsequently, the stability of the closed-loop system is analyzed.

\subsection{Hyperplane-Based Sliding Manifold}
A hyperplane-based sliding variable is constructed by integrating Integral Terminal Sliding Mode Control (ITSMC)~\citep{sun2021trajectory} and Nonsingular Terminal Sliding Mode Control (NTSMC)~\citep{xie2021robust} techniques to achieve precise, fast, and robust tracking performance.

Define the tracking errors as
\[
e_x = \Gamma_{11} - \Gamma_{xd}, \qquad e_y = \Gamma_{21} - \Gamma_{yd},
\]
and their derivatives as
\[
\dot{e}_x = \dot{\Gamma}_{11} - \dot{\Gamma}_{xd}, \qquad \dot{e}_y = \dot{\Gamma}_{21} - \dot{\Gamma}_{yd},
\]
where $\Gamma_{xd}$ and $\Gamma_{yd}$ denote the reference trajectories for $\Gamma_{11}$ and $\Gamma_{21}$, respectively.

The ITSM for the positions along the $x$ and $y$ axes is defined as
\begin{equation}\label{eq:n4a}
s_x = \kappa_{x1} e_x + \kappa_{x2} \int |e_x|^{\Phi_x} \operatorname{sign}(e_x)\, dt,
\end{equation}
\begin{equation}\label{eq:n4b}
s_y = \kappa_{y1} e_y + \kappa_{y2} \int |e_y|^{\Phi_y} \operatorname{sign}(e_y)\, dt,
\end{equation}
where $\kappa_{i1},\kappa_{i2} \in \mathbb{R}^+$ and $\Phi_i\in(0.5,1)$.

Following your modification, the hyperplane-based sliding manifolds are chosen using the derivative-based term:
\begin{equation}\label{eq:n4c_dot}
\sigma_x = s_x + \mu_x |\dot{s}_x|^{\beta_x}\operatorname{sign}(\dot{s}_x),
\end{equation}
\begin{equation}\label{eq:n4d_dot}
\sigma_y = s_y + \mu_y |\dot{s}_y|^{\beta_y}\operatorname{sign}(\dot{s}_y),
\end{equation}
where $\mu_i>0$ and $\beta_i>1$ are design parameters.

\subsection{Control Law Design}
Differentiate (\ref{eq:n4c_dot})--(\ref{eq:n4d_dot}). For $i\in\{x,y\}$ we get
\begin{equation}\label{eq:dot_sigma}
\dot{\sigma}_i = \dot{s}_i + \mu_i \beta_i |\dot{s}_i|^{\beta_i-1} \ddot{s}_i,
\end{equation}
valid for $\dot{s}_i\neq 0$ (and interpreted in the Filippov sense at $\dot{s}_i=0$).

The controller is composed of an equivalent term and a switching term:
\begin{equation}\label{eq:vi_split}
v_i = v_{eqi} + v_{swi},
\end{equation}
where, following the same design philosophy as before, one may take for instance
\begin{equation}\label{eq:veqi}
v_{eqi} = \frac{1}{\kappa_{i1}}\Big(\dot{\Gamma}_{id} - \kappa_{i2}|e_i|^{\Phi_i}\operatorname{sign}(e_i)\Big),
\end{equation}
and
\begin{equation}\label{eq:vswi}
v_{swi} = -\frac{1}{\kappa_{i1}}\Big(\Upsilon_{i1}\sigma_i + \Upsilon_{i2}\operatorname{sign}(\sigma_i)\Big),
\end{equation}
with $\Upsilon_{i1},\Upsilon_{i2}>0$.

Using the system relations (i.e., $\ddot{\Gamma}_{11},\ddot{\Gamma}_{21}$ expressed via $v_x,v_y$ and then mapped to actuator commands) the closed-loop substitution yields an expression of $\ddot{s}_i$ containing $v_i$, bounded model terms and perturbations which we denote compactly by $\varpi_i$.

\subsection{Finite-time reaching and sliding theorems}

\begin{theorem}\label{th:reach_sigma}
Consider the closed-loop system under the control law \eqref{eq:vi_split}--\eqref{eq:vswi} and the hyperplane-based manifolds
\begin{equation*}
\sigma_i = s_i + \mu_i |\dot{s}_i|^{\beta_i}\operatorname{sign}(\dot{s}_i),\qquad \mu_i>0,\ \beta_i>1.
\end{equation*}
Assume the perturbations are bounded and satisfy $|\kappa_{i1}\varpi_i|\le\Upsilon_{i2}$. Then for any initial condition with $\dot{s}_i\neq0$ almost everywhere there exists a finite time $t_{f_i}>0$ such that
\[
\sigma_i(t)=0,\qquad \forall t\ge t_{f_i}.
\]
That is, each sliding variable $\sigma_i$ is reached in finite time.
\end{theorem}

\begin{proof}
Take the Lyapunov candidate $V_\Sigma=\tfrac12(\sigma_x^2+\sigma_y^2)$. Using \eqref{eq:dot_sigma} (and substituting the closed-loop expression for $\ddot s_i$ into the compact perturbation term $\kappa_{i1}\varpi_i$) we may write, for $\dot s_i\neq0$,
\[
\dot\sigma_i=\mu_i\beta_i|\dot s_i|^{\beta_i-1}\big(-\Upsilon_{i1}\sigma_i-\Upsilon_{i2}\operatorname{sign}(\sigma_i)+\kappa_{i1}\varpi_i\big).
\]
Hence
\[
\dot V_\Sigma=\sum_{i=x,y}\sigma_i\mu_i\beta_i|\dot s_i|^{\beta_i-1}\big(-\Upsilon_{i1}\sigma_i-\Upsilon_{i2}\operatorname{sign}(\sigma_i)+\kappa_{i1}\varpi_i\big).
\]
By the assumption $|\kappa_{i1}\varpi_i|\le\Upsilon_{i2}$ the discontinuous term is dominated and we obtain
\[
\dot V_\Sigma \le -\sum_{i=x,y}\mu_i\beta_i\,|\dot s_i|^{\beta_i-1}\Upsilon_{i1}\sigma_i^2 \le 0,
\]
so $V_\Sigma$ is nonincreasing and $\sigma_i$ are bounded.

To show finite-time reaching introduce the time-scale transformation
\[
d\tau_i=\mu_i\beta_i\,|\dot s_i|^{\beta_i-1}dt,
\]
which is strictly increasing in $t$ for trajectories with $\dot s_i\neq0$ a.e. Dividing the scalar dynamics by the positive multiplier yields, in the $\tau_i$ time,
\[
\frac{d\sigma_i}{d\tau_i}=-\Upsilon_{i1}\sigma_i-\Upsilon_{i2}\operatorname{sign}(\sigma_i)+\kappa_{i1}\varpi_i.
\]
With $|\kappa_{i1}\varpi_i|<\Upsilon_{i2}$ the right-hand side is a (uniformly) strictly negative feedback outside a neighborhood of zero and standard finite-time arguments for first-order linear dynamics plus a discontinuous offset imply existence of a finite $\tau_{r_i}>0$ such that $\sigma_i(\tau_i)=0$ for all $\tau_i\ge\tau_{r_i}$. Because the map $t\mapsto\tau_i(t)$ is continuous and strictly increasing, the corresponding real time $t_{f_i}$ (the preimage of $\tau_{r_i}$) is finite. Thus $\sigma_i$ is reached in finite real time and the trajectory thereafter remains on the manifold by Filippov-invariance arguments.
\end{proof}

\begin{theorem}\label{th:sliding_cascade}
Under the same assumptions as Theorem~\ref{th:reach_sigma}, once $\sigma_i(t)=0$ for $t\ge t_{f_i}$ the reduced-order (sliding) dynamics satisfy the following cascade of finite-time convergences:
\begin{enumerate}
  \item The auxiliary variable $z_i:=\dot s_i$ obeys
  \[
  \dot z_i = -\frac{1}{\mu_i\beta_i}\,|z_i|^{2-\beta_i}\operatorname{sign}(z_i),
  \]
  and reaches zero in finite time
  \begin{equation}\label{eq:Th2_Tz}
  T_{z,i} \;=\; \frac{\mu_i\beta_i}{\beta_i-1}\,|z_i(t_{f_i})|^{\,\beta_i-1}.
  \end{equation}
  \item Consequently $s_i$ vanishes in finite time (by the algebraic manifold relation) and, after $z_i\equiv0$, the tracking error $e_i$ evolves as
  \[
  \dot e_i = -\frac{\kappa_{i2}}{\kappa_{i1}}|e_i|^{\Phi_i}\operatorname{sign}(e_i),
  \]
  which reaches zero in finite time
  \begin{equation}\label{eq:Th2_Te}
  T_{e,i} \;=\; \frac{\kappa_{i1}}{(1-\Phi_i)\kappa_{i2}}\,|e_i(t_{f_i}+T_{z,i})|^{\,1-\Phi_i}.
  \end{equation}
\end{enumerate}
Thus $z_i\to0$, $s_i\to0$ and $e_i\to0$ in finite time along the sliding manifold.
\end{theorem}

\begin{proof}
On the manifold $\sigma_i=0$ the algebraic relation 
\[
s_i + \mu_i|z_i|^{\beta_i}\operatorname{sign}(z_i)=0
\]
holds, where $z_i=\dot s_i$. For intervals where $z_i\neq0$ and keeps its sign, differentiate the relation and use the identity \(\frac{d}{dt}(|z|^{\beta-1}z)=\beta|z|^{\beta-1}\dot z\) to obtain
\[
z_i = -\mu_i\beta_i |z_i|^{\beta_i-1}\dot z_i.
\]
Rearranging yields the scalar dynamics
\[
\dot z_i = -\frac{1}{\mu_i\beta_i}\,|z_i|^{1-\beta_i} z_i
= -\frac{1}{\mu_i\beta_i}\,|z_i|^{2-\beta_i}\operatorname{sign}(z_i).
\]
Because $\beta_i>1$ we have exponent $2-\beta_i\in(0,1)$ and the above is a finite-time stable scalar ODE. Separation of variables and direct integration give the settling time \eqref{eq:Th2_Tz} for $z_i$ to reach zero starting from $z_i(t_{f_i})$.

From the manifold algebraic relation $s_i=-\mu_i|z_i|^{\beta_i-1}z_i$ it follows that $s_i(t)\to0$ as $z_i(t)\to0$, hence $s_i$ vanishes in finite time no later than $t_{f_i}+T_{z,i}$.

After $z_i\equiv0$ we have $\dot s_i=0$, and using $\dot s_i=\kappa_{i1}\dot e_i+\kappa_{i2}|e_i|^{\Phi_i}\operatorname{sign}(e_i)$ yields
\[
0=\kappa_{i1}\dot e_i+\kappa_{i2}|e_i|^{\Phi_i}\operatorname{sign}(e_i),
\]
which is the finite-time scalar dynamics for $e_i$ with exponent $\Phi_i\in(0,1)$. Integrating this equation gives the settling time \eqref{eq:Th2_Te} from the value $e_i(t_{f_i}+T_{z,i})$. This completes the cascade finite-time convergence proof.
\end{proof}

\begin{remark}
\begin{itemize}
  \item Theorems~\ref{th:reach_sigma} and \ref{th:sliding_cascade} together provide a full picture: first the controller forces $\sigma_i$ to zero in finite time (reaching phase), then while sliding the internal variables $\dot s_i$, $s_i$ and $e_i$ converge to zero in finite time (sliding phase).
  \item Filippov solutions are implicitly used to handle sign changes and the switching term at the manifold; a fully rigorous Filippov treatment can be added if desired.
  \item The expressions \eqref{eq:Th2_Tz} and \eqref{eq:Th2_Te} give explicit dependence of settling times on design parameters and initial conditions on the manifold.
\end{itemize}
\end{remark}

\section{Experimental Validation Results}

To exhibit the practical viability of the proposed controller, a real indoor experiment was conducted on the WMR. To illustrate the improved performance of the proposed RTC law, the flatness-based sliding mode control (FBSMC) technique was used for comparison. The parameters of the proposed RTC law are chosen as outlined: $\kappa_{x1}=3,\kappa_{x2}=0.1, \kappa_{y1}=3, \kappa_{y2}=0.1, \Phi_x=0.95, \Phi_y=0.95, \mu_x=1.14, \mu_y=1.14, \beta_x=1.28, \beta_y=1.28, \Upsilon_{x1}=0.04, \Upsilon_{x2}=0.02, \Upsilon_{y1}=0.04$ and $\Upsilon_{y2}=0.02$.  Fig.~\ref{fig:setup} displays the experimental setup and configuration. A router established a shared WiFi network between the WMR and the laptop, with the laptop acting as the master. The WMR's motion was controlled via the Robot Operating System (ROS) on the laptop. Soft foam boards were placed on the floor to introduce disturbances that could impact TT performance. Additionally, a blower was used to simulate wind gusts, introducing environmental factors that affect the WMR's performance and emulate real-world conditions. Videos of the hands-on experiment are available in~\footnote{\url{https://youtu.be/jGpGLuDYfdw?si=lB8erlgZuPFvyAW5}}. 

To test robustness against strong external disturbances and to analyze each controller's potential to handle the TT problem, a wind from the blower was applied around t=70 sec as a strong disturbance while following the referenced trajectory. Fig.~\ref{fig:result1} depicts the TT performance of both controllers. To provide a meaningful comparison, the tracking errors along $x$ and $y$ are shown in Fig.~\ref{fig:result5} and Fig.~\ref{fig:result6}. The WMR exhibits a noticeable error when exposed to severe wind gusts from the blower. However, it subsequently converges and follows the desired path. Importantly, the actual WMR trajectory deviated slightly from the desired route. This discrepancy may result from a number of factors, including track conditions, wheel dynamics, communication delays, and the size and shape of the WMR. However, the tracking accuracy remains within an acceptable range for practical applications. Furthermore, selecting a large positive value of $\mu_i$ and a lower value of $\beta_i$ can accelerate the convergence of $e_x$ and $e_y$, however, it also increases the overall magnitude of the control input. Therefore, the optimal gain selection is displayed here. In addition, Figures~\ref{fig:result7} and \ref{fig:result8} depict the translational and rotational velocities. The WMR can achieve maximum rotational and translational velocities of 2.84 rad/s and 0.22 m/s, respectively. Notably, both controllers were tested with a saturation limit imposed on both velocities to ensure compliance with the maximum allowable limits. The actual velocities recorded for the proposed controller during the experiments remained within these constraints. Under strong disturbances, the SMC demonstrated robustness while attempting to exceed these limits. To give a more comprehensive perspective, a quantitative study was done, focusing the superior TT capability of the proposed strategy. The Integral Absolute Error $(IAE)$, Integral Square Error $(ISE)$ along $x$ and $y$, and the average control utilization $(P_{avg})$ were used as evaluation metrices in the experiment, as illustrated in Table~\ref{table:1}.  The proposed method exhibits low control energy usage while preserving effective TT capability, which is crucial in environments with limited energy availability.

\begin{figure}[]
	\centering 
\includegraphics[width=0.7\linewidth,keepaspectratio]{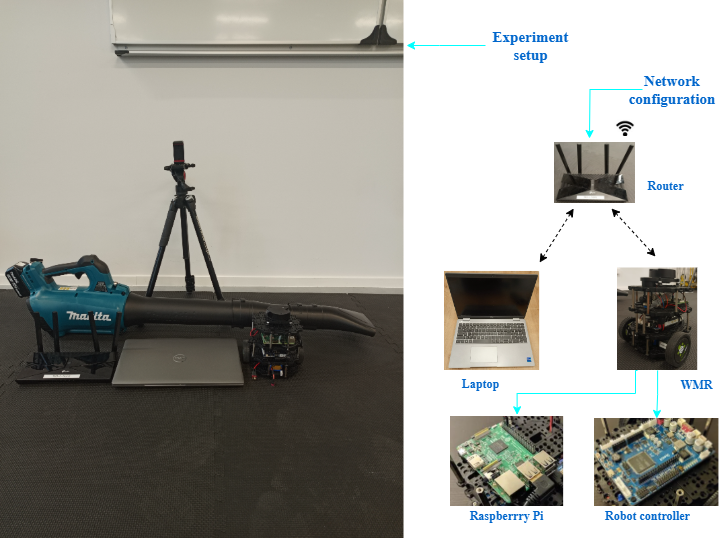}
	\caption{Experimental setup.}\label{fig:setup}
\end{figure}
\begin{figure}[]
    \centering
    \includegraphics[scale=0.53]{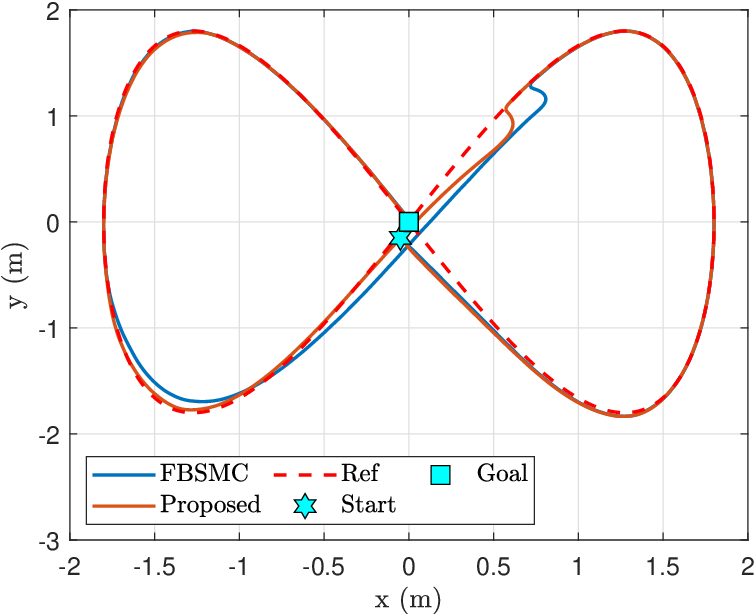}
    \caption{Control and referenced trajectories.}
    \label{fig:result1}
\end{figure}

\begin{figure}[]
    \centering
    \includegraphics[scale=0.58]{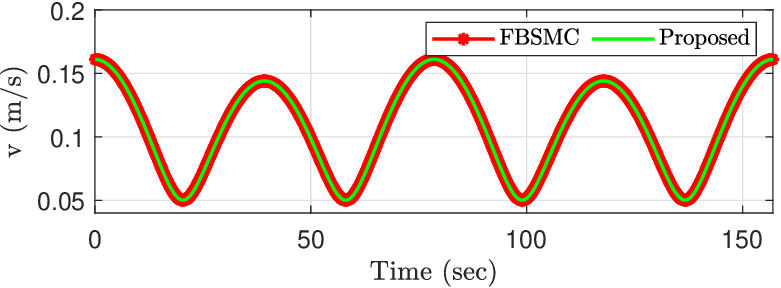}
    \caption{Tracking error of x.}
    \label{fig:result5}
\end{figure}
\begin{figure}[]
    \centering
    \includegraphics[scale=0.58]{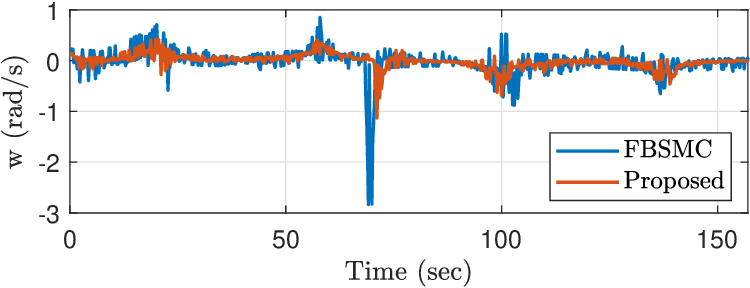}
    \caption{Tracking error of y.}
    \label{fig:result6}
\end{figure}
\begin{figure}[]
    \centering
    \includegraphics[scale=0.58]{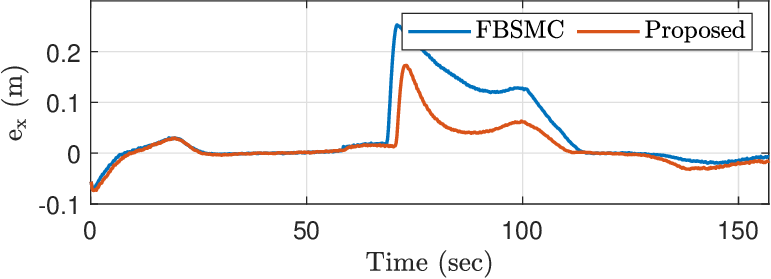}
    \caption{Linear velocity.}
    \label{fig:result7}
\end{figure}
\begin{figure}[]
    \centering
    \includegraphics[scale=0.58]{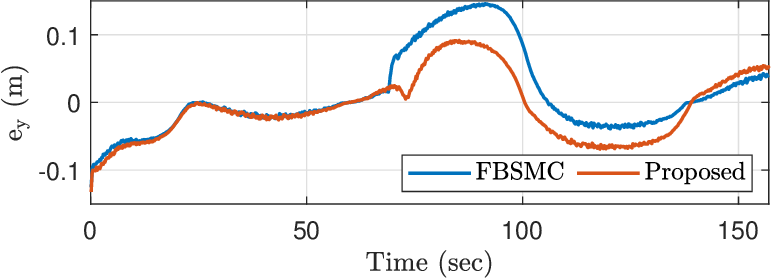}
    \caption{Angular velocity.}
    \label{fig:result8}
\end{figure}

\begin{table}[]
\centering
    \caption{Quantitative evaluation}
    \begin{tabular}{ccccccc}
        \hline
         \textbf{Controller} & \textbf{IAE$_{x}$} & \textbf{IAE$_{y}$} & \textbf{ISE$_{x}$} & \textbf{ISE$_{y}$} & \textbf{P$_{\text{avg}}$} \\
        \hline
        FBSMC &  7.0131 &  6.9040 &  0.9719 &  0.5829 &  0.0967   \\
        Proposed &  3.8904  & 6.4075 &  0.2465 &  0.3763 &  0.0293  \\ \hline
    \end{tabular}%
\label{table:1}
\end{table}

\section{Conclusions}

This article is dedicated to designing a novel robust control algorithm for precise tracking of WMRs subjected to strong external perturbations. First, the WMR model is transformed into a linear canonical form by exploiting the DF of its KM. Then, in the control design process, two nonlinear sliding surface variables are embedded to construct new sliding manifolds for the WMR. Furthermore, discontinuous control laws are formulated to Counteract the impact of strong disturbances, ensuring robustness, with stability formally proven. Additionally, comparative physical experiments on the TurtleBot3 WMR under strong disturbances, confirm the effectiveness of the proposed algorithm both quantitatively 
 and qualitatively thereby, showcasing its practical applicability. Future research includes extending
the proposed strategy to mobile manipulators
and quadrotors to evaluate its effectiveness and explore
its potential for wider applications.

\bibliography{bib}             

@inproceedings{yang2018sliding,
  title={A Sliding mode control method for trajectory tracking control of wheeled mobile robot},
  author={Yang, Lu and Pan, Shenghui},
  booktitle={Journal of Physics: Conference Series},
  volume={1074},
  number={1},
  pages={012059},
  year={2018},
  organization={IOP Publishing}
}

@article{nath2021event,
  title={Event-triggered sliding-mode control of two wheeled mobile robot: an experimental validation},
  author={Nath, Krishanu and Yesmin, Asifa and Nanda, Anirban and Bera, Manas Kumar},
  journal={IEEE Journal of Emerging and Selected Topics in Industrial Electronics},
  volume={2},
  number={3},
  pages={218--226},
  year={2021},
  publisher={IEEE}
}

@article{mu2017nonlinear,
  title={Nonlinear sliding mode control of a two-wheeled mobile robot system},
  author={Mu, Jianqiu and Yan, Xing-Gang and Spurgeon, Sarah K and Mao, Zehui},
  journal={International Journal of Modelling, Identification and Control},
  volume={27},
  number={2},
  pages={75--83},
  year={2017},
  publisher={Inderscience Publishers (IEL)}
}

@article{wu2019backstepping,
  title={Backstepping trajectory tracking based on fuzzy sliding mode control for differential mobile robots},
  author={Wu, Xing and Jin, Peng and Zou, Ting and Qi, Zeyu and Xiao, Haining and Lou, Peihuang},
  journal={Journal of Intelligent \& Robotic Systems},
  volume={96},
  pages={109--121},
  year={2019},
  publisher={Springer}
}

@article{dang2023adaptive,
  title={Adaptive backstepping hierarchical sliding mode control for 3-wheeled mobile robots based on RBF neural networks},
  author={Dang, Son Tung and Dinh, Xuan Minh and Kim, Thai Dinh and Xuan, Hai Le and Ha, Manh-Hung},
  journal={Electronics},
  volume={12},
  number={11},
  pages={2345},
  year={2023},
  publisher={MDPI}
}

@article{singhal2022robust,
  title={Robust trajectory tracking control of non-holonomic wheeled mobile robots using an adaptive fractional order parallel fuzzy PID controller},
  author={Singhal, Kartik and Kumar, Vineet and Rana, KPS},
  journal={Journal of the Franklin Institute},
  volume={359},
  number={9},
  pages={4160--4215},
  year={2022},
  publisher={Elsevier}
}

@inproceedings{benaziza2017mobile,
  title={Mobile robot trajectory tracking using terminal sliding mode control},
  author={Benaziza, Walid and Slimane, Noureddine and Mallem, Ali},
  booktitle={2017 6th International Conference on Systems and Control (ICSC)},
  pages={538--542},
  year={2017},
  organization={IEEE}
}

@article{hassan2024path,
  title={Path planning and trajectory tracking control for two-wheel mobile robot},
  author={Hassan, Ibrahim A and Abed, Issa A and Al-Hussaibi, Walid A},
  journal={Journal of Robotics and Control (JRC)},
  volume={5},
  number={1},
  pages={1--15},
  year={2024}
}

@article{ebel2024cooperative,
  title={Cooperative object transportation with differential-drive mobile robots: Control and experimentation},
  author={Ebel, Henrik and Rosenfelder, Mario and Eberhard, Peter},
  journal={Robotics and Autonomous Systems},
  volume={173},
  pages={104612},
  year={2024},
  publisher={Elsevier}
}

@article{yepez2023mobile,
  title={Mobile robotics in smart farming: current trends and applications},
  author={Y{\'e}pez-Ponce, Dar{\'\i}o Fernando and Salcedo, Jos{\'e} Vicente and Rosero-Montalvo, Pa{\'u}l D and Sanchis, Javier},
  journal={Frontiers in Artificial Intelligence},
  volume={6},
  pages={1213330},
  year={2023},
  publisher={Frontiers Media SA}
}

@article{garaffa2021reinforcement,
  title={Reinforcement learning for mobile robotics exploration: A survey},
  author={Garaffa, Lu{\'\i}za Caetano and Basso, Maik and Konzen, Andr{\'e}a Aparecida and de Freitas, Edison Pignaton},
  journal={IEEE Transactions on Neural Networks and Learning Systems},
  volume={34},
  number={8},
  pages={3796--3810},
  year={2021},
  publisher={IEEE}
}

@article{zhao2021wheeled,
  title={A wheeled robot chain control system for underground facilities inspection using visible light communication and solar panel receivers},
  author={Zhao, Wen and Kamezaki, Mitsuhiro and Yamaguchi, Kaoru and Konno, Minoru and Onuki, Akihiko and Sugano, Shigeki},
  journal={IEEE/ASME Transactions on Mechatronics},
  volume={27},
  number={1},
  pages={180--189},
  year={2021},
  publisher={IEEE}
}

@article{al2021embedded,
  title={Embedded design and implementation of mobile robot for surveillance applications},
  author={Al, Abdulkareem Sh Mahdi Al-Obaidi and Al-Qassa, Arif and Nasser, Ahmed R and Alkhayyat, Ahmed and Humaidi, Amjad J and Ibraheem, Ibraheem K and others},
  journal={Indonesian Journal of Science and Technology},
  volume={6},
  number={2},
  pages={427--440},
  year={2021}
}

@article{zangina2020non,
  title={Non-linear PID controller for trajectory tracking of a differential drive mobile robot},
  author={Zangina, Umar and Buyamin, Salinda and Abidin, Mohamad Shukri Zainal and Azimi, Mohd Saiful and Hasan, HS},
  journal={Journal of Mechanical Engineering Research and Developments},
  volume={43},
  number={7},
  pages={255--270},
  year={2020}
}

@article{yousuf2021dynamic,
  title={Dynamic modeling and tracking for nonholonomic mobile robot using PID and back-stepping},
  author={Yousuf, Bilal M and Saboor Khan, Abdul and Munir Khan, Sana},
  journal={Advanced Control for Applications: Engineering and Industrial Systems},
  volume={3},
  number={3},
  pages={e71},
  year={2021},
  publisher={Wiley Online Library}
}

@article{xu2020enhanced,
  title={Enhanced bioinspired backstepping control for a mobile robot with unscented Kalman filter},
  author={Xu, Zhe and Yang, Simon X and Gadsden, S Andrew},
  journal={IEEE Access},
  volume={8},
  pages={125899--125908},
  year={2020},
  publisher={IEEE}
}

@article{tang2010differential,
  title={Differential-flatness-based planning and control of a wheeled mobile manipulator—Theory and experiment},
  author={Tang, Chin Pei and Miller, Patrick T and Krovi, Venkat N and Ryu, Ji-Chul and Agrawal, Sunil K},
  journal={IEEE/ASME Transactions on Mechatronics},
  volume={16},
  number={4},
  pages={768--773},
  year={2010},
  publisher={IEEE}
}

@article{xie2021robust,
  title={Robust tracking control of a differential drive wheeled mobile robot using fast nonsingular terminal sliding mode},
  author={Xie, Hao and Zheng, Jinchuan and Chai, Rifai and Nguyen, Hung T},
  journal={Computers \& Electrical Engineering},
  volume={96},
  pages={107488},
  year={2021},
  publisher={Elsevier}
}

@article{sun2021trajectory,
  title={Trajectory-tracking control of Mecanum-wheeled omnidirectional mobile robots using adaptive integral terminal sliding mode},
  author={Sun, Zhe and Hu, Shujie and He, Defeng and Zhu, Wei and Xie, Hao and Zheng, Jinchuan},
  journal={Computers \& Electrical Engineering},
  volume={96},
  pages={107500},
  year={2021},
  publisher={Elsevier}
}

@article{korayem2024adaptive,
  title={Adaptive robust control with slipping parameters estimation based on intelligent learning for wheeled mobile robot},
  author={Korayem, MH and Safarbali, M and Lademakhi, N Yousefi},
  journal={ISA transactions},
  volume={147},
  pages={577--589},
  year={2024},
  publisher={Elsevier}
}

@article{tolossa2024trajectory,
  title={Trajectory tracking control of a mobile robot using fuzzy logic controller with optimal parameters},
  author={Tolossa, Tesfaye Deme and Gunasekaran, Manavaalan and Halder, Kaushik and Verma, Hitendra Kumar and Parswal, Shyam Sundar and Jorwal, Nishant and Joseph, Felix Orlando Maria and Hote, Yogesh Vijay},
  journal={Robotica},
  pages={1--24},
  year={2024},
  publisher={Cambridge University Press}
}

@article{chen2023nonlinear,
  title={Nonlinear adaptive fuzzy control design for wheeled mobile robots with using the skew symmetrical property},
  author={Chen, Yung-Hsiang and Chen, Yung-Yue},
  journal={Symmetry},
  volume={15},
  number={1},
  pages={221},
  year={2023},
  publisher={MDPI}
}

@article{li2018adaptive,
  title={Adaptive neural network tracking control-based reinforcement learning for wheeled mobile robots with skidding and slipping},
  author={Li, Shu and Ding, Liang and Gao, Haibo and Chen, Chao and Liu, Zhen and Deng, Zongquan},
  journal={Neurocomputing},
  volume={283},
  pages={20--30},
  year={2018},
  publisher={Elsevier}
}

@article{chwa2011fuzzy,
  title={Fuzzy adaptive tracking control of wheeled mobile robots with state-dependent kinematic and dynamic disturbances},
  author={Chwa, Dongkyoung},
  journal={IEEE transactions on Fuzzy Systems},
  volume={20},
  number={3},
  pages={587--593},
  year={2011},
  publisher={IEEE}
}

@article{moudoud2022fuzzy,
  title={Fuzzy adaptive sliding mode controller for electrically driven wheeled mobile robot for trajectory tracking task},
  author={Moudoud, Brahim and Aissaoui, Hicham and Diany, Mohammed},
  journal={Journal of Control and Decision},
  volume={9},
  number={1},
  pages={71--79},
  year={2022},
  publisher={Taylor \& Francis}
}

@article{chen2020adaptive,
  title={Adaptive-neural-network-based trajectory tracking control for a nonholonomic wheeled mobile robot with velocity constraints},
  author={Chen, Ziyu and Liu, Yang and He, Wei and Qiao, Hong and Ji, Haibo},
  journal={IEEE Transactions on Industrial Electronics},
  volume={68},
  number={6},
  pages={5057--5067},
  year={2020},
  publisher={IEEE}
}

@article{liu2023autonomous,
  title={Autonomous planning and robust control for wheeled mobile robot with slippage disturbances based on differential flat},
  author={Liu, Yueyue and Bai, Keqiang and Wang, Haoyu and Fan, Qigao},
  journal={IET Control Theory \& Applications},
  volume={17},
  number={16},
  pages={2136--2145},
  year={2023},
  publisher={Wiley Online Library}
}

@article{mera2020sliding,
  title={A sliding-mode based controller for trajectory tracking of perturbed unicycle mobile robots},
  author={Mera, Manuel and R{\'\i}os, H{\'e}ctor and Mart{\'\i}nez, Edgar A},
  journal={Control Engineering Practice},
  volume={102},
  pages={104548},
  year={2020},
  publisher={Elsevier}
}

@article{yuan2024differential,
  title={Differential flatness-based adaptive robust tracking control for wheeled mobile robots with slippage disturbances},
  author={Yuan, Wang and Liu, Yueyue and Liu, Yong-Hua and Su, Chun-Yi},
  journal={ISA transactions},
  volume={144},
  pages={482--489},
  year={2024},
  publisher={Elsevier}
}

@article{fliess1995flatness,
  title={Flatness and defect of non-linear systems: introductory theory and examples},
  author={Fliess, Michel and L{\'e}vine, Jean and Martin, Philippe and Rouchon, Pierre},
  journal={International journal of control},
  volume={61},
  number={6},
  pages={1327--1361},
  year={1995},
  publisher={Taylor \& Francis}
}
 
\end{document}